\documentclass{ifacconf}
\usepackage{amsmath}
\usepackage{algorithm}
\usepackage{algpseudocode}
\usepackage{natbib}
\newenvironment{proof}[1][Proof]{\par\noindent\textbf{#1.} }{\hfill$\square$\par}

\usepackage{amsfonts}

\usepackage{subcaption}

\newtheorem{remark}{Remark}
\newtheorem{lemma}{Lemma}
\newtheorem{algorithm2}{Algorithm}
\usepackage{xcolor}
\usepackage{tikz}
\definecolor{dblue}{rgb}{0 0 0.7}
\definecolor{red}{rgb}{1 0 0}
\definecolor{MatlabBlue}{rgb}{0     , 0.4470, 0.7410}
\definecolor{MatlabRed}{rgb}{0.6350 0.0780 0.1840}
\definecolor{MatlabOrange}{rgb}{0.8500, 0.3250, 0.0980}
\definecolor{MatlabYellow}{rgb}{0.9290, 0.6940, 0.1250}
\definecolor{MatlabPurple}{rgb}{0.4940, 0.1840, 0.5560}
\definecolor{MatlabGreen}{rgb}{0.4660, 0.6740, 0.1880}
\definecolor{MatlabBabyBlue}{rgb}{0.3010, 0.7450, 0.9330}
\definecolor{MatlabGray}{rgb}{0.5, 0.5, 0.5}
\definecolor{MatlabLightGray}{rgb}{0.75, 0.75, 0.75}
\definecolor{MatlabBlack}{rgb}{0, 0, 0}
\definecolor{MatlabLightGray4}{rgb}{0.875, 0.875, 0.875}
\definecolor{MatlabLightGray3}{rgb}{0.85, 0.85, 0.85}
\definecolor{MatlabLightGray2}{rgb}{0.775, 0.775, 0.775}
\definecolor{MatlabLightGray1}{rgb}{0.7, 0.7, 0.7}
\definecolor{MatlabGray20}{rgb}{0.2, 0.2, 0.2}
\definecolor{MatlabGray30}{rgb}{0.3, 0.3, 0.3}
\definecolor{MatlabGray40}{rgb}{0.4, 0.4, 0.4}
\definecolor{MatlabGray50}{rgb}{0.5, 0.5, 0.5}
\definecolor{MatlabGray60}{rgb}{0.6, 0.6, 0.6}
\definecolor{MatlabGray70}{rgb}{0.7, 0.7, 0.7}
\definecolor{MatlabGray80}{rgb}{0.8, 0.8, 0.8}
\definecolor{MatlabGray85}{rgb}{0.85, 0.85, 0.85}
\definecolor{MatlabGray90}{rgb}{0.9, 0.9, 0.9}
 
\newcommand{\tikzline}[1]{(\protect\tikz[baseline=-0.6ex,x=1pt,y=1pt, line width=0.4mm]{ \protect\draw[#1] [-] (0,0) -- (10,0);})}

\newcommand{\tikzmark}[1]{(\protect\tikz[baseline=-0.6ex,x=1pt,y=1pt]{ \protect\draw[color=#1, fill=#1] (5,0) circle (1pt)})}

\newcommand{\hop}{\mathsf{H}}

\usepackage{graphicx}      
\usepackage{natbib}        
\setcitestyle{doi=false} 

\usepackage{eso-pic}
\AddToShipoutPictureBG*{%
\AtPageUpperLeft{%
\setlength\unitlength{1in}%
\hspace*{\dimexpr0.5\paperwidth\relax}
\makebox(0,-1.75)[c]{
\begin{tabular}{c c}
Maarten van der Hulst \emph{et al.},
Frequency domain identification for multivariable motion control systems: \\
Applied to a prototype wafer stage, 
uploaded to ArXiv  \today \\
\end{tabular}}}}

\begin{document}
\begin{frontmatter}

\title{Frequency domain identification for multivariable motion control systems:
Applied to a prototype wafer stage} 

\thanks[footnoteinfo]{This project is funded by Holland High Tech | TKI HSTM via the PPP Innovation Scheme (PPP-I) for public-private partnerships.}

\author[First]{M. van der Hulst}
\author[First]{R. A. Gonz\'alez}
\author[First]{K. Classens}
\author[First]{P. Tacx}
\author[Second]{N. Dirkx}
\author[Second]{J. van de Wijdeven}
\author[First,Third]{T. Oomen}

\address[First]{Dept. of Mechanical Engineering, Eindhoven University of Technology, The Netherlands}
\address[Second]{ASML, Veldhoven, The Netherlands}
\address[Third]{Delft Center for Systems and Control, Delft University of Technology The Netherlands}

\begin{abstract} 
Multivariable parametric models are essential for optimizing the performance of high-tech systems. The main objective of this paper is to develop an identification strategy that provides accurate parametric models for complex multivariable systems. To achieve this, an additive model structure is adopted, offering advantages over traditional black-box model structures when considering physical systems. The introduced method minimizes a weighted least-squares criterion and uses an iterative linear regression algorithm to solve the estimation problem, achieving local optimality upon convergence. Experimental validation is conducted on a prototype wafer-stage system, featuring a large number of spatially distributed actuators and sensors and exhibiting complex flexible dynamic behavior, to evaluate performance and demonstrate the effectiveness of the proposed method.
\end{abstract}

\begin{keyword}
Parameter estimation, system identification, multivariable systems, frequency
response function
\end{keyword}

\end{frontmatter}

\section{Introduction}
System identification involves developing mathematical models using experimental data, often incorporating insights from physical principles \citep{Ljung1999SystemUser}. Data-driven parametric models of multi-input multi-output (MIMO) systems are essential to optimize performance of engineered systems as they enable the design of high-performance controllers and observers, provide design validation and feedback, and facilitate online monitoring and fault diagnosis \citep{Steinbuch2022MotionLaw}.

Traditional linear system identification approaches for multivariable systems often rely on black-box model structures that do not consider the underlying structure of the considered physical system. Examples include rational common denominator models and matrix fractional descriptions (MFDs) \citep{Pintelon2012SystemIdentification}. The literature on these model parameterizations is extensive \citep{Glover1974ParametrizationsIdentifiability, Correa1984Pseudo-canonicalModels, Vayssettes2016NewApproach}, yet may not provide the most parsimonious or physically relevant model descriptions when considering practical applications. Many physical systems are more naturally described by a sum of low-order transfer functions. Examples are found in vibrational analysis \citep{Vayssettes2015NewStructures, Zhang2019VariationalEngineering, Dorosti2018IterativeApproach} and control of flexible motion systems \citep{Voorhoeve2021IdentifyingStage, Tacx2024IdentificationOutputs}, where models are often represented as a sum of transfer functions with distinct denominators, corresponding to the individual resonant modes of the system \citep{Gawronski2004AdvancedStructures}. Similar approaches are found in thermal analysis of machine frames \citep{Zhu2008RobustConcept}, RLC circuits \citep{Lange2021BroadbandShape} and acoustic modeling of room responses \citep{Jian2022AcousticInterpolation}. The estimation of additive transfer function models, which are related to unfactored transfer functions by a partial fraction expansion, offers several advantages. These models enable more efficient parameterization by minimizing the number of parameters needed to represent the system, thereby reducing model complexity and enhancing the statistical estimation properties \citep{Soderstrom2001SystemIdentification}.  Furthermore, they provide enhanced physical insight for fault diagnosis \citep{Classens2022FaultResonances} and improve numerical conditioning, which is crucial for the parametric identification of stiff and high-order systems \citep{Gilson2018AIdentification}.

When identifying physical systems, estimating continuous-time models offers distinct advantages over discrete-time models. Continuous-time models facilitate the integration of \textit{a priori} knowledge, such as relative degree, and provide more interpretable parameters which directly correspond to physical quantities \citep{Garnier2015DirectApplications}. Herein, the frequency-domain approach for the parametric identification of continuous-time models has become increasingly popular. Frequency-domain system identification offers several advantages, including data and computational efficiency, flexible data processing, nonparametric noise model estimation, and direct interpretation of system dynamics \citep{Pintelon2012SystemIdentification}. 

Many MIMO frequency-domain identification strategies have been developed for models parametrized in non-additive structures. These methods can largely be categorized into pseudo-linear regression-based approaches \citep{Blom2010MultivariableRegression, Sanathanan1961TransferPolynomials} and gradient descent methods \citep{Bayard1994High-orderResults}. In contrast, the estimation of additive model parametrizations has primarily been explored in single-input single-output (SISO) approaches. One such method is vector fitting \citep{Semlyen1999RationalFitting}, which considers fitting first-order pole models. Recent advancements in additive system identification include the direct continuous-time identification method introduced in \cite{Gonzalez2024IdentificationClosed-loop}, which is based on the simplified refined instrumental variable method (SRIVC) \cite{Young1980RefinedAnalysis}, as well as a block coordinate descent approach with variants for both offline and online parameter estimation \cite{Gonzalez2023ParsimoniousApproach, Classens2024RecursiveSystems}.

Although additive identification offers several advantages for estimating models of physical systems, most existing methods focus on the SISO setting, while many practical applications require a MIMO formulation. This paper aims to introduce a comprehensive identification method for estimating additive linear continuous-time MIMO systems using frequency-domain data. The main contributions of this paper are:

\begin{itemize}
    \item[C1] A frequency-domain refined instrumental variable method for estimating continuous-time MIMO systems in additive transfer function form.
    \item[C2] Experimental validation of the developed identification on a prototype wafer-stage system. 
\end{itemize}
This paper is organized as follows. First, Section 2 formally introduces the additive model structure and outlines the identification problem considered. In section 3, the identification strategy is presented with experimental validation in Section 4. Finally,  conclusions are given in Section 5.

\textit{Notation:} Scalars, vectors and matrices are written as $x$, $\mathbf{x}$ and $\mathbf{X}$, respectively. The imaginary unit is denoted by $ j^2 = -1 $, and for $\mathbf{z}\in\mathbb{C}^n$, the operation $\Re\{\mathbf{z}\}$ returns the real part of the complex vector $\mathbf{z}$. For a matrix $ \mathbf{A} $, its transpose is written as $ \mathbf{A}^{\top} $, and its Hermitian (conjugate transpose) as $ \mathbf{A}^{\hop} $. If $ \mathbf{x} \in \mathbb{C}^n $ and $ \mathbf{Q} \in \mathbb{C}^{n \times n} $ is a Hermitian matrix, then the weighted 2-norm is given by $ \|\mathbf{x}\|_{\mathbf{Q}} = \sqrt{\mathbf{x}^{\hop} \mathbf{Q} \mathbf{x}} $. For $ \mathbf{X} = [\mathbf{x}_1, \ldots, \mathbf{x}_n] $, with $ \mathbf{x}_i \in \mathbb{C}^n $, the operation $ \operatorname{vec}(\mathbf{X}) = [\mathbf{x}^\top_1, \ldots, \mathbf{x}^\top_n]^\top $ restructures the matrix into a vector by stacking its columns.

\section{Setup and Problem formulation}
In this section, the experimental setup is presented and the additive model structure is formally introduced. Finally, the identification problem considered is formulated.

\subsection{Experimental setup: prototype wafer-stage system}
The considered experimental setup depicted in Figure \ref{fig: oat_overview} is a prototype wafer stage system.  The system is actively controlled in six motion degrees of freedom at a sampling rate of 10 kHz, achieving accuracy in the sub-micrometer range. The stage is magnetically levitated using gravity compensators, achieving a mid-air equilibrium and eliminating any mechanical connections to the fixed world. These systems exhibit pronounced flexible dynamics, which pose significant challenges for controller design, model updating, design feedback, and monitoring techniques. The availability of accurate mathematical models that capture the dynamics of the flexible multivariable system is crucial for effectively addressing these challenges.

The system contains 17 actuators: 13 are in the $z$ direction and two each for the $x$ and $y$ directions. Furthermore, the system includes 7 sensors, 4 for the $z$-direction, 2 for the $x$-directions, and a single sensor for the $y$-direction.  Only out-of-plane motions, that is, translation along the $z$ axis, and rotations around the $x$ and $y$ axes, are considered in this paper to facilitate the exposition. An overview of the sensors and actuators considered is provided in Figure \ref{fig: oat_schematic}.

\begin{figure}[t]
    \centering
    \includegraphics[width=0.85\columnwidth]{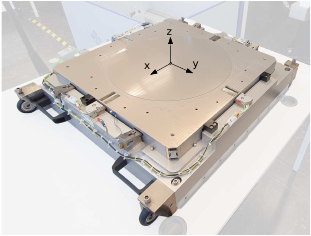}
    \caption{Experimental setup featuring a prototype wafer-stage system.}
    \label{fig: oat_overview}

    \centering
    \includegraphics{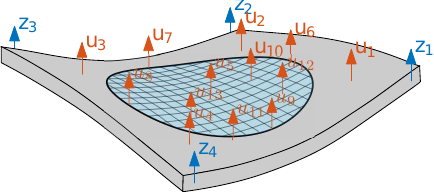}
    \caption{Schematic overview of the featured actuators $u_i$ and sensors $z_i$ in the out-of-plane direction.}
    \label{fig: oat_schematic}
\end{figure}

\subsection{Model structure} \label{sec: additive model setup}
To model the input/output dynamics of the wafer-stage system, an additive model structure is adopted, which is formally introduced in this section. Consider the linear and time-invariant (LTI) model of a MIMO system with $n_{\mathrm{u}}$ inputs and $n_{\mathrm{y}}$ outputs in additive form 
\begin{equation} \label{eq: FREQ - true plant}
\mathbf{P}(s,\boldsymbol{\beta}) = \sum_{i=1}^K\mathbf{P}_i(s,\boldsymbol{\theta}_i),
\end{equation}
with $K$ the number of submodels, $s$ the Laplace variable and $\boldsymbol{\beta}$ and $\boldsymbol{\theta}_i$ the joint and submodel parameter vector.  Each submodel $\mathbf{P}_i(s,\boldsymbol{\theta}_i)$ is parametrized according to
\begin{equation}
\mathbf{P}_i(s,\boldsymbol{\beta}) = \frac{1}{ s^{\ell_i} A_i(s)}\mathbf{B}_i(s),
\end{equation} 
where at most one submodel may include $\ell_i>0$ poles at the origin. The scalar denominator polynomial $A_i(s)$ and the matrix numerator polynomial $\mathbf{B}_i(s)$ are such that no complex number $z$ simultaneously satisfies $A_i(z) = 0$ and $ \mathbf{B}_i(z) = \mathbf{0}$. To ensure a unique characterization of $ \left\{\mathbf{P}_i(s)\right\}_{i=1}^K $, it is assumed that at most one submodel $ \mathbf{P}_i(s) $ is biproper.  The $A_i(s)$ and $\mathbf{B}_i(s)$ polynomials are parametrized as
\begin{align} 
A_i(s) &= 1 + a_{i,1}s + \ldots + a_{i,n_i}s^{n_i}, \label{eq: FREQ - plant polynomial A}\\
\mathbf{B}_i(s) &= \mathbf{B}_{i,0} + \mathbf{B}_{i,1} s+ \ldots + \mathbf{B}_{i,m_i}s^{m_i}, \label{eq: FREQ - plant polynomial B}
\end{align}
where the $A_i(s)$ polynomials are stable, i.e., all roots lie in the left-half plane, and they not share any common roots. The polynomials $A_i(s)$ and $\mathbf{B}_i(s)$ are jointly described by the parameter vector
\begin{align} \label{eq: FREQ - parameter vector}
\boldsymbol{\beta} = \left[ \boldsymbol{\theta}^{\top}_1,\ \dots,\ \boldsymbol{\theta}^{\top}_K \right]^{\top},
\end{align}
where $\boldsymbol{\theta}_i$ for $i = 1,\ldots,K$ contains the parameters of the $i$th submodel
\begin{align} \label{eq: FREQ - parameter vector submodel}
\hspace{-0.75em} \boldsymbol{\theta}_i = \left[ a_{i,1}, \dots, a_{i,n_i}, \operatorname{vec}\left(\mathbf{B}_{i,0}\right)^{\top}, \dots, \operatorname{vec}\left(\mathbf{B}_{i,m_i}\right)^{\top}\right]^{\top} \hspace{-0.5em}.
\end{align}

\subsection{Identification problem}
A dataset of noisy plant frequency response function (FRF) measurements $\mathbf{G}(\omega_k) \in \mathbb{C}^{n_{\mathrm{y}} \times n_{\mathrm{u}}}$ of length $N$, is assumed to be available for the estimation problem. To estimate continuous-time models from the measured FRF, the pseudo-continuous-time setting is adopted (see \cite{Wang2008IdentificationData}, Chapter 8 for details).  The identification problem is formulated based on the matrix residual, which is computed as the difference between the FRF measurement and the model, according to
\begin{equation} \label{eq: FREQ - residual}
\mathbf{E}\left(\omega_k, \boldsymbol{\beta}\right) = \mathbf{G}\left(\omega_k\right) - \mathbf{P}\left(\xi_k, \boldsymbol{\beta} \right),
\end{equation}
where $\xi_k = j\omega_k$. The parameter vector estimate $\hat{\boldsymbol{\beta}}$ is obtained as the minimizer of the weighted least-squares criterion
\begin{equation} \label{eq: FREQ - optimization problem}
\hat{\boldsymbol{\beta}}=\underset{\boldsymbol{\beta}}{\arg \min }\frac{1}{2N}\sum_{k=1}^N\left\| \operatorname{vec}\bigl(\mathbf{E}\left(\omega_k, \boldsymbol{\beta}\right)\bigr) \right\|_{\mathbf{W}(\omega_k)}^2,
\end{equation}
where $\mathbf{W}(\omega_k) \in \mathbb{C}^{n_{\mathrm{u}}n_{\mathrm{y}} \times n_{\mathrm{u}}n_{\mathrm{y}}}$ is a frequency-dependent weighting matrix. The problem considered is to estimate additive models as described by \eqref{eq: FREQ - true plant}, that
minimizes the cost function in \eqref{eq: FREQ - optimization problem}, given a dataset of noisy FRF measurements.

\section{Identification strategy for additive MIMO systems}
In this section, an iterative linear regression method is introduced to solve the nonlinear and nonconvex optimization problem \eqref{eq: FREQ - optimization problem}, thereby constituting contribution C1.

\subsection{Criterion for optimality}
The minimizers of the cost function in \eqref{eq: FREQ - optimization problem} satisfy the first-order optimality condition 
\begin{equation} \label{eq: FREQ - first-order optimality}
\begin{aligned}
\mathbf{0} = \frac{1}{N}\sum_{k=1}^N \Re \left\{ \hat{\mathbf{\Phi}}\left(\omega_k, \boldsymbol{\beta}\right) \mathbf{W}(\omega_k) \operatorname{vec}\bigl(\mathbf{E}\left(\omega_k, \boldsymbol{\beta}\right)\bigr) \right\},
\end{aligned}
\end{equation}
with the gradient 
\begin{equation}
\hat{\mathbf{\Phi}}\left(\omega_k, \boldsymbol{\beta}\right) = \left(\frac{\partial \operatorname{vec}\bigl(\mathbf{E}\left(\omega_k, \boldsymbol{\beta}\right)\bigr)}{\partial \boldsymbol{\beta}^\top}\right)^\hop.     
\end{equation}
For the considered additive model structure the gradient corresponds to
\begin{equation} \label{eq: FREQ - instrument matrix}
    \hat{\mathbf{\Phi}}\left(\omega_k, \boldsymbol{\beta}\right) = \Bigl[ \begin{array}{ccc} \hat{\mathbf{\Phi}}_{1}^\hop\left(\omega_k, \boldsymbol{\theta}_1\right) & \ldots & \hat{\mathbf{\Phi}}_{K}^\hop\left(\omega_k, \boldsymbol{\theta}_K\right)\end{array}\Bigr]^\hop, 
\end{equation}
where $\hat{\mathbf{\Phi}}_{i}\left(\omega_k, \boldsymbol{\theta}_i\right)$ for $i = 1,\ldots,K$ is given by
\begin{equation} \label{eq: FREQ - instrument matrix submodel}
\begin{aligned} 
\hat{\boldsymbol{\Phi}}_{i}(\omega_k,\boldsymbol{\beta}) = 
\Bigg[
& \frac{-\xi_k \mathbf{p}_i(\xi_k, \boldsymbol{\theta}_i)}{\xi_k^{\ell_i}A_i(\xi_k)}, 
\ldots, 
\frac{-\xi_k^{n_i} \mathbf{p}_{i}(\xi_k, \boldsymbol{\theta}_i)}{\xi_k^{\ell_i}A_i(\xi_k)}, \\
& \frac{\mathbf{I}_{n_{\mathrm{u}}n_{\mathrm{y}}}}{\xi_k^{\ell_i}A_i(\xi_k)}, 
\ldots, 
\frac{\xi_k^{m_i} \mathbf{I}_{n_{\mathrm{u}}n_{\mathrm{y}}}}{\xi_k^{\ell_i}A_i(\xi_k)} 
\Bigg]^\hop,
\end{aligned}
\end{equation}
with $\mathbf{p}_i(\xi_k, \boldsymbol{\theta}_i) = \operatorname{vec}(\mathbf{P}_i(\xi_k, \boldsymbol{\theta}_i))$ the vectorized plant of the $i$th submodel.  In the following subsections, the first-order optimality condition \eqref{eq: FREQ - first-order optimality} will be exploited to derive an estimator for the parameter vector $\hat{\boldsymbol{\beta}}$.

\subsection{Refined instrumental variables for additive systems}
The condition in (\ref{eq: FREQ - first-order optimality}) is non-linear in the parameter vector $\boldsymbol{\beta}$. A solution is obtained by reformulating \eqref{eq: FREQ - residual} to a pseudolinear form which enables the refined instrumental variables approach. For each submodule in the additive model structure, the residual can be reformulated into an unique pseudolinear regression, as stated in the following lemma. \\

\begin{lemma} The pseudolinear regression form of the residual \eqref{eq: FREQ - residual} corresponding to the $i$th submodel is expressed as
\begin{equation} \label{eq: FREQ - subproblems regression form}
\operatorname{vec}\bigl(\mathbf{E}\left(\omega_k, \boldsymbol{\beta}\right)\bigr)  = \tilde{\mathbf{g}}_{f,i}(\omega_k, \boldsymbol{\beta}) -\boldsymbol{\Phi}_{i}^{\top}\left(\omega_k, \boldsymbol{\beta}\right) \boldsymbol{\theta}_i, 
\end{equation}
with the regressor
\begin{equation}
\begin{aligned}
\boldsymbol{\Phi}_i(\omega_k,\boldsymbol{\beta}) = 
\Bigg[ 
&\frac{-\xi_k \tilde{\mathbf{g}}_{i}(\omega_k, \boldsymbol{\beta})}{A_i(\xi_k)} , 
\ldots, 
\frac{-\xi_k^{n_i}\tilde{\mathbf{g}}_{i}(\omega_k, \boldsymbol{\beta})}{A_i(\xi_k)} , \\
&\frac{\mathbf{I}_{n_{\mathrm{u}}n_{\mathrm{y}}}}{\xi^{\ell_i}A_i(\xi_k)}, 
\ldots, 
\frac{\xi_k^{m_i}\mathbf{I}_{n_{\mathrm{u}}n_{\mathrm{y}}}}{\xi_k^{\ell_i}A_i(\xi_k)} 
\Bigg]^\top,
\end{aligned}
\end{equation}
and where $\tilde{\mathbf{g}}_{f,i}(\omega_k, \boldsymbol{\beta}) = A^{-1}_i(\xi_k)\tilde{\mathbf{g}}_{i}(\omega_k, \boldsymbol{\beta})$ with $\tilde{\mathbf{g}}_i(\omega_k, \boldsymbol{\beta}) = \operatorname{vec}(\tilde{\mathbf{G}}_i(\omega_k, \boldsymbol{\beta}))$ the  residual plant of the $i$th submodel, defined by
\begin{equation} \label{eq: FREQ - residual plant}
\tilde{\mathbf{G}}^{}_{i}\left(\omega_k, \boldsymbol{\beta}\right) = \mathbf{G}^{}(\omega_k) - \sum_{\substack{\ell=1, \ldots, K \\ \ell \neq i}} \mathbf{P}_\ell(\xi_k, \boldsymbol{\theta}_\ell).
\end{equation}
\end{lemma}
\begin{proof} \textit{
The residual \eqref{eq: FREQ - residual} is rewritten for $i = 1,\ldots,K$ according to
\begin{align}
\mathbf{E}\left(\omega_k, \boldsymbol{\beta} \right) & =\tilde{\mathbf{G}}_i\left(\omega_k,\boldsymbol{\beta}\right)-\frac{\mathbf{B}_i\left(\xi_k\right)}{\xi_k^{\ell_i}A_i\left(\xi_k\right)}, \\
& \hspace{-3em}=\frac{1}{\xi_k^{\ell_i}A_i\left(\xi_k\right)}\left(\xi_k^{\ell_i}A_i\left(\xi_k\right) \tilde{\mathbf{G}}_i\left(\omega_k,\boldsymbol{\beta}\right)-\mathbf{B}_i\left(\xi_k\right)\right), \label{eq: FREQ - subproblems}
\end{align}
with $\tilde{\mathbf{G}}_i$ defined in \eqref{eq: FREQ - residual plant}. Substituting the numerator and denominator polynomials \eqref{eq: FREQ - plant polynomial A} and \eqref{eq: FREQ - plant polynomial B}, and vectorizing both sides, \eqref{eq: FREQ - subproblems} yields
\begin{align}
\operatorname{vec}\bigl(\mathbf{E}\left(\omega_k, \boldsymbol{\beta}\right)\bigr) &=  \frac{\tilde{\mathbf{g}}_i\left(\omega_k, \boldsymbol{\beta}\right)}{A_i\left(\xi_k\right)} + \ldots + \frac{a_{n_i}\xi_k^{n_i} \tilde{\mathbf{g}}_i\left(\omega_k, \boldsymbol{\beta}\right)}{A_i\left(\xi_k\right)} \notag \\
& \hspace{-0em} -\frac{\operatorname{vec}\left(\mathbf{B}_{i, 0}\right)}{\xi_k^{\ell_i}A_i\left(\xi_k\right)}-\ldots-\frac{\xi_k^{m_i} \operatorname{vec}\left(\mathbf{B}_{i, m}\right)}{\xi_k^{\ell_i}A_i\left(\xi_k\right)}.
\end{align}
This expression can directly be written in the form \eqref{eq: FREQ - subproblems regression form} by considering \eqref{eq: FREQ - parameter vector submodel}, thereby completing the proof.}
\end{proof}

The residual formulation in \eqref{eq: FREQ - subproblems regression form} defines $K$ pseudolinear regressions. Introducing the stacked signals
\begin{align}
    \mathbf{\Upsilon}\left(\omega_k, \boldsymbol{\beta}\right) &= \Bigl[ \begin{array}{ccc}  \tilde{\mathbf{g}}_{f,1}(\omega_k,\boldsymbol{\beta}) & \dots & \tilde{\mathbf{g}}_{f,K}(\omega_k,\boldsymbol{\beta})\end{array}\Bigl]^\top, \label{eq: FREQ - stacked residual output} \\ 
    \mathbf{\Phi}\left(\omega_k, \boldsymbol{\beta}\right) &= \Bigl[ \begin{array}{ccc} \mathbf{\Phi}^\top_{1}\left(\omega_k, \boldsymbol{\beta}\right) & \dots & \mathbf{\Phi}^\top_{K}\left(\omega_k, \boldsymbol{\beta}\right)\end{array}\Bigr]^\top \label{eq: FREQ - stacked regressor matrix},
\end{align}
and the parameter matrix
\begin{equation} \label{eq: FREQ - parameter matrix}
\mathcal{B}=\left[\begin{array}{ccc}
\boldsymbol{\theta}_1 & & \mathbf{0} \\
& \ddots & \\
\mathbf{0} & & \boldsymbol{\theta}_K
\end{array}\right],
\end{equation}
which contains the elements of $\boldsymbol{\beta}$ along the block diagonal, allows to write the equivalent optimality condition \eqref{eq: FREQ - first-order optimality} for the $K$ subproblem as
\begin{equation} \label{eq: FREQ - instrumental variable equation - matrix}
\begin{aligned}
    \sum_{k=1}^N \Re \left\{\hat{\mathbf{\Phi}}(\omega_k, \boldsymbol{\beta})\mathbf{W}(\omega_k) \Bigl( \mathbf{\Upsilon}^\top(\omega_k, \boldsymbol{\beta}) - \mathbf{\Phi}^\top(\omega_k,\boldsymbol{\beta})\mathcal{B} \Bigr) \right\} = \mathbf{0}.
\end{aligned}
\end{equation}
The solution to \eqref{eq: FREQ - instrumental variable equation - matrix} is found iteratively by fixing $\boldsymbol{\beta} = \boldsymbol{\beta}^{\langle j \rangle}$ at the $j$th iteration in (\ref{eq: FREQ - stacked residual output}), the regressor (\ref{eq: FREQ - stacked regressor matrix}), and additionally the gradient (\ref{eq: FREQ - instrument matrix}), which leads to the following iterative procedure.
\begin{algorithm2} Given an initial estimate $\boldsymbol{\beta}^{\langle 0\rangle}$ and maximum number of iterations $M$, compute a new estimate by iterating
\begin{align} \label{eq: FREQ - estimator}
\hat{\mathcal{B}}^{\langle j+1 \rangle} = {\left[ \sum_{k=1}^N \hat{\mathbf{\Phi}}(\omega_k, \boldsymbol{\beta}^{\langle j \rangle}) \mathbf{W}(\omega_k)\mathbf{\Phi}^{\top}(\omega_k, \boldsymbol{\beta}^{\langle j \rangle})\right]^{-1} }  \times  \notag \\
\sum_{k=1}^N \hat{\mathbf{\Phi}}(\omega_k, \boldsymbol{\beta}^{\langle j \rangle}) \mathbf{W}(\omega_k)\mathbf{\Upsilon}^{\top}(\omega_k, \boldsymbol{\beta}^{\langle j \rangle}),
\end{align}
where the next iteration $\boldsymbol{\beta}^{j+1}$ is extracted from the block diagonal coefficients of $\mathcal{B}^{j+1}$ as in (\ref{eq: FREQ - parameter matrix}).
\end{algorithm2}
The convergence point of the iterations described by (\ref{eq: FREQ - estimator}) provides a solution to the first-order optimality condition (\ref{eq: FREQ - first-order optimality}). This ensures that the estimate corresponds to a stationary point of the cost function in (\ref{eq: FREQ - optimization problem}), thereby guaranteeing (local) optimality.\\

\begin{remark} Note that the iterations described by \eqref{eq: FREQ - estimator} corresponds to a refined instrumental variable method, where $\hat{\mathbf{\Phi}}$ is interpreted as the instrument matrix \citep{Young1980RefinedAnalysis}. Furthermore, for $K=1$, the iterations in (\ref{eq: FREQ - estimator}) correspond to the frequency-domain refined instrumental variable method in \cite{Blom2010MultivariableRegression}, and by replacing $\hat{\mathbf{\Phi}}$ with $\mathbf{\Phi}$ to the SK iterations by \cite{Sanathanan1961TransferPolynomials}. The method can be considered as a frequency-domain variant of the approach introduced in \cite{Gonzalez2024IdentificationClosed-loop} for MIMO systems.
\end{remark}

\subsection{Initialization}
The iterations in \eqref{eq: FREQ - estimator} require an initial estimate $\boldsymbol{\beta}^{\langle 0 \rangle}$ of the model parameters. This section introduces a method for computing the numerator parameters assuming fixed denominator polynomials. This
reduces the initialization problem to determining initial pole
locations, which are often effectively obtained from, e.g., finite element models or nonparametric FRF models. To this end, assume that the denominator polynomials are fixed at $\bar{A}_i(s)$, and let $\boldsymbol{\eta}$ represent the parameter vector from \eqref{eq: FREQ - parameter vector} without the denominator coefficients. The estimate $\hat{\boldsymbol{\eta}}$ is found as the solution to the convex problem
\begin{equation} \label{eq: convex problem}
    \hat{\boldsymbol{\eta}} = \underset{\boldsymbol{\eta}}{\arg \min }\frac{1}{2N}\sum^N_{k=1}\|\operatorname{vec}\bigl(\mathbf{G}(\omega_k)\bigr)-\boldsymbol{\Phi}^{\top}\left(\omega_k\right) \boldsymbol{\eta}\|_{2}^2,
\end{equation}
where the regressor matrix $\boldsymbol{\Phi}$ is obtained by stacking for each submodel
\begin{align} 
\mathbf{\Phi}_{i}\left(\omega_k\right) =&\ \left[ \dfrac{\mathbf{I}_{n_{\mathrm{u}}n_{\mathrm{y}}}}{\xi^{\ell_i}\bar{A}_i(\xi_k)},\dots,\ \dfrac{\xi^{m_i}\mathbf{I}_{n_{\mathrm{u}}n_{\mathrm{y}}}}{\xi^{\ell_i}\bar{A}_i(\xi_k)} \right]^{\top},
\end{align}
in the same way as \eqref{eq: FREQ - stacked regressor matrix}. Hence, an initial estimate $\boldsymbol{\beta}^{\langle 0 \rangle}$ is determined by first providing initial pole locations, which enable the computation of the numerator parameters by solving the convex problem \eqref{eq: convex problem} given data.

\section{Experimental validation}
This section presents the experimental validation of the introduced identification strategy, thereby providing contribution C2. The considered system is the prototype wafer-stage system introduced in Section 2.

\subsection{Model structure}
The input/output dynamics of the wafer-stage system containing $n_{\mathrm{y}}=13$ outputs and $n_{\mathrm{u}}=4$ inputs, is modeled in the additive structure
\begin{equation} \label{eq: FREQ - model structure}
    \mathbf{P}(s, \boldsymbol{\beta}) =  \frac{\mathbf{B}_{i,0}}{s^2}  + \sum_{i=1}^{n_{\mathrm{flex}}} \frac{\mathbf{B}_{i,0}}{s^2/ \omega_{i}^{2} + 2(\zeta_i/\omega_i)s + 1},
\end{equation}
where the components of the decomposition are interpreted as rigid-body modes and flexible dynamic modes, with $\omega_i$ the resonance frequencies, $\zeta_i$ the corresponding damping coefficients, and $n_{\mathrm{flex}}$ the number of flexible modes \citep{Gawronski2004AdvancedStructures}. 

\subsection{Nonparametric modeling} As a first step in frequency-domain identification, a nonparametric model needs to be identified. The nonparametric FRF model of the $(4 \times 13)$ plant, representing the out-of-plane dynamics of the wafer-stage system, is obtained using the multi-sine approach described in \cite{Pintelon2012SystemIdentification}. The experiments are performed in a closed-loop configuration since active control of the mid-air equilibrium is required for stable operation. The plant FRF is derived using the indirect method, where the system is excited by $n_{\mathrm{u}} $ single-axis random-phase multisine signals with a flat amplitude spectrum. The multi-sine excitation includes $ 10 $ periods and $ 10 $ realizations, resulting in a plant FRF consisting of $ N = 4000 $ complex data points spanning a frequency range of 0.25 Hz to 2000 Hz. Frequency lines below 20 Hz are discarded during the parametric estimation step, as the rigid-body behavior is poorly captured at lower frequencies in the measurement. The delays introduced by the hold circuit in the digital measurement environment are determined based on the FRF model. The dataset is then compensated for these delays, allowing the delay-corrected FRF to be modeled in continuous time \cite[Chapter 8]{Wang2008IdentificationData}.

\subsection{Weighting filter design}
For the weighting filter an element-wise inverse plant magnitude weighting is selected, given by 
\begin{equation}
    \mathbf{W}(\omega_k) = \operatorname{diag}\Bigl(\operatorname{vec}\left(|\mathbf{G}(\omega_k)|\right)\Bigr)^{-1}.
\end{equation}
The inverse plant magnitude weighting effectively transforms the matrix residual \eqref{eq: FREQ - residual} from absolute to relative error criterion. This prevents overemphasizing frequencies with a large magnitude, which can dominate the estimation process, especially for systems containing integrator dynamics. 

\subsection{Initialization and parametric identification}
To determine the number of flexible modes, $n_{\mathrm{flex}}$, and their corresponding frequency locations, the Complex Mode Indicator Function (CMIF) is used, as outlined in \cite{Shih1988ComplexEstimation}. The CMIF is computed as the square of the singular values of the FRF matrix evaluated at each frequency point. A mode is indicated by a peak in the CMIF, with the frequency location of the peak corresponding to the damped natural frequency of the flexible mode. Since the setup is a lightly damped system, the damped natural frequency is approximately equal to the natural frequency, and therefore provides for an accurate initial estimate of the frequency location. 

In Figure \ref{fig: cmif}, the CMIF of the FRF dataset is provided. Using this approach, $n_{\mathrm{flex}}=17$ distinct flexible modes are found in the frequency range considered. The frequency locations of the peaks in the CMIF are used to initialize the natural frequencies $\omega_i$. The corresponding damping coefficients are initialized at $\zeta_i = 0.01$ for $i = 1,\ldots, n_{\mathrm{flex}}$, which are typical values encountered for lightly damped poles in these systems. The initial modal parameters determine the pole locations, which enables the numerator parameters to be computed using the convex problem \eqref{eq: convex problem}. Finally, the initial parameter vector $\boldsymbol{\beta}^{\langle0\rangle}$ is constructed as in (\ref{eq: FREQ - parameter vector}) and is used to initialize the iterations described by \eqref{eq: FREQ - estimator}. 

\subsection{Results}
The frequency response of the estimated plant model is shown in Figure \ref{fig: oat_fit}, together with the nonparametric FRF measurement used in the estimation. Furthermore, Figure \ref{fig: oat_fit single} shows the frequency response of a single plant entry, along with the corresponding residual. The parametric model accurately aligns with the FRF measurement over the complete frequency range, demonstrating the validity of the introduced method. In particular, the high-frequent flexible modes are accurately modeled in the additive structure, which can be challenging to achieve using traditional model structures.

\begin{figure}
    \centering
    \includegraphics[width=\columnwidth]{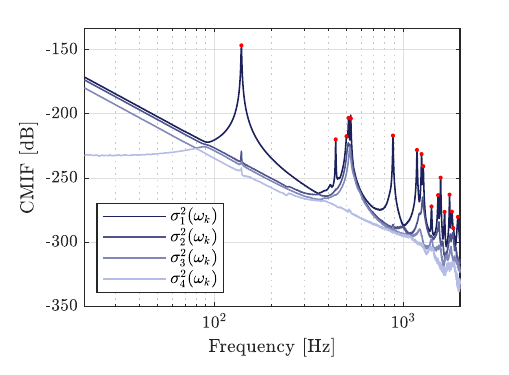}
    \caption{CMIF plot with $\sigma_i(\omega_k)$ the $i$th singular value of the FRF and \tikzmark{red} indicating the selected modes.}
    \label{fig: cmif}
\end{figure}

\begin{figure*}
    \centering    \includegraphics[width=\linewidth]{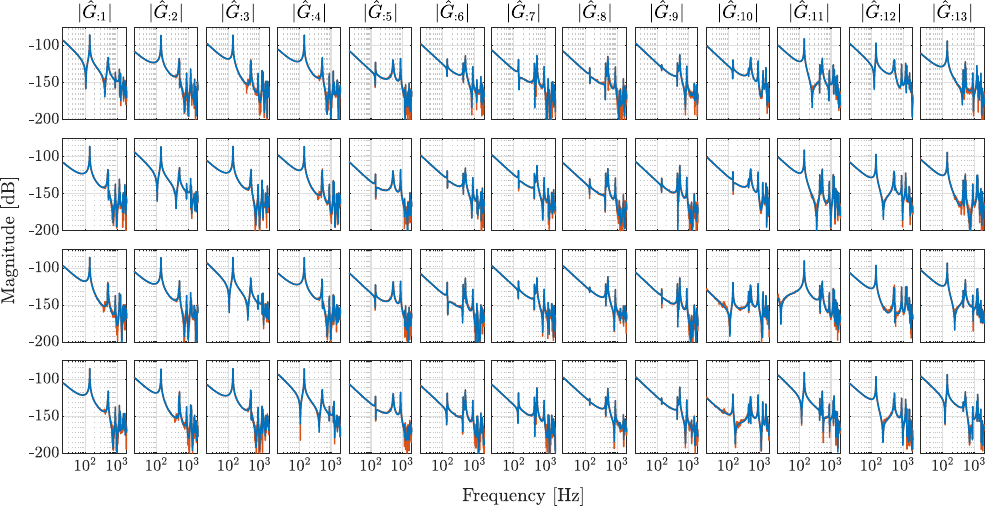}
    \caption{Element-wise Bode magnitude plot of the the nonparametric FRF measurement \tikzline{MatlabOrange} and the estimated parametric model \tikzline{MatlabBlue}.}
    \label{fig: oat_fit}
\end{figure*}

\begin{figure}
    \centering
    \includegraphics[width=\columnwidth]{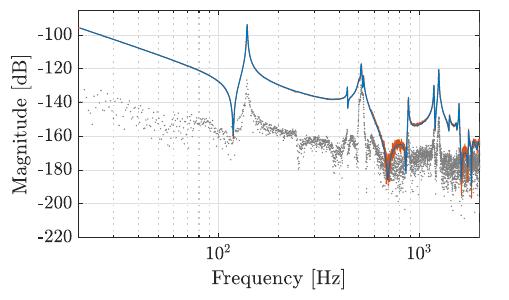}
    \caption{Bode magnitude plot of the $\hat{G}_{4,13}$ entry, with the nonparametric FRF measurement \tikzline{MatlabOrange}, the estimated parametric model \tikzline{MatlabBlue} and residual \tikzmark{MatlabGray}. }
    \label{fig: oat_fit single}
\end{figure}

\section{Conclusion}
This paper addresses the parametric identification of multivariable systems using frequency-domain datasets. 
The introduced method, which uses an iterative regression algorithm to minimize a least-squares criterion, enables direct estimation of additive transfer function models. Many systems are more naturally described in an additive structure, leading to reduced complexity models, improved conditioning, and enhanced physical insight. The procedure has been successfully tested on a prototype wafer-stage system, providing accurate models over a large frequency range.

\renewcommand{\bibfont}{\footnotesize} 

\bibliography{references}             
                                                   






\end{document}